\documentclass[journal,onecolumn]{IEEEtran}

\usepackage[T1]{fontenc}
\usepackage{amsmath,amssymb,amsfonts,amsthm,mathtools}
\usepackage{bm}
\usepackage{array}
\usepackage{booktabs}
\usepackage{enumitem}
\usepackage{cite}
\usepackage{url}
\usepackage{microtype}
\allowdisplaybreaks
\newtheorem{theorem}{Theorem}
\newtheorem{lemma}{Lemma}

\newtheorem{corollary}{Corollary}

\newtheorem{remark}{Remark}

\newcommand{\F}{\mathbb F}
\newcommand{\C}{\mathcal C}
\newcommand{\D}{\mathcal D}
\newcommand{\B}{\mathcal B}
\newcommand{\N}{\mathcal N}
\newcommand{\Q}{\mathcal Q}
\newcommand{\Oset}{\mathcal O}
\newcommand{\U}{\mathcal U}
\newcommand{\PG}{\operatorname{PG}}
\newcommand{\PGL}{\operatorname{PGL}}
\newcommand{\GL}{\operatorname{GL}}
\newcommand{\Tr}{\operatorname{Tr}}
\newcommand{\Norm}{\operatorname{N}}
\newcommand{\Span}{\operatorname{span}}
\newcommand{\rank}{\operatorname{rank}}
\newcommand{\ord}{\operatorname{ord}}
\newcommand{\wt}{\operatorname{wt}}
\newcommand{\supp}{\operatorname{supp}}
\newcommand{\Aut}{\operatorname{Aut}}

\newcommand{\eps}{\varepsilon}
\newcommand{\sig}{\sigma}

\providecommand{\lcm}{\operatorname{lcm}}

\begin{document}

\title{Algebraic Resolutions of Seven Open Problems on Cyclic and Negacyclic Codes Supporting Designs}

\author{Yutong~Zhang,~\IEEEmembership{Member,~IEEE,}
        and~Yaoran~Yang,~\IEEEmembership{Member,~IEEE}%
\thanks{Manuscript received Month Day, 2026; revised Month Day, 2026.}
\thanks{Yutong Zhang and Yaoran Yang are with the School of Mathematics,
Sichuan University, Chengdu 610065, China
(e-mail: yutongzhang@stu.scu.edu.cn; yangyaoran@stu.scu.edu.cn).}
\thanks{Corresponding author: Yutong Zhang.}}

\markboth{IEEE Transactions on Information Theory,~Vol.~XX, No.~XX, Month~2026}%
{Zhang and Yang: Algebraic Resolutions of Seven Open Problems on Cyclic and Negacyclic Codes Supporting Designs}

\maketitle

\begin{abstract}
This paper gives a unified algebraic solution to seven open problems of Wang, Tang and Ding on cyclic, negacyclic and constacyclic codes supporting designs.  For the cyclic code
\[
C\left(\frac{p^s-1}{2},\frac{p^s+1}{2}\right),
\]
a Cayley parametrization of the unit circle reduces the trace-zero condition to a semilinear equation on \(\PG(1,q)\).  Its large root sets are exactly the \(\F_{p^{\gcd(m,s)}}\)-sublines, yielding the complementary design
\[
\overline{S(3,q_0+1,q+1)}.
\]
For the length \(q^2+1\) negacyclic code, a quotient transport from \(\U_{2(q^2+1)}\) to \(\U_{q^2+1}\) and a unit-circle parametrization show that the minimum zero sets are precisely the Baer sublines of \(\PG(1,q^2)\).  Equivalently, the corresponding support design is the complement of the non-tangent plane sections of an elliptic quadric \(\Q^-(3,q)\).

For constacyclic ovoid codes of length \(q^2+1\) over \(\F_q\), the exact existence criterion is
\[
\lambda\in\F_q^*,\qquad
\exists\ \lambda\text{-constacyclic ovoid code}
\Longleftrightarrow
\lambda\notin(\F_q^*)^2.
\]
In particular, negacyclic ovoid codes exist exactly when \(q\equiv3\pmod4\).  The proof uses the corrected projective-order congruence
\[
a=(q+1)c,\qquad c\equiv b\pmod{q-1},\qquad
\operatorname{ord}(\theta\F_q^*)=\frac{q^2+1}{\gcd(q^2+1,c)}.
\]
The paper also derives a universal weight enumerator for lifted ovoid codes over extension fields, independent of the chosen ovoid.  Finally, consecutive-root negacyclic MDS codes are constructed to give complete simple \(5\)-designs, including a proper negacyclic \([11,5,7]_{23}\) code whose minimum supports form the complete \(5-(11,7,15)\) design.
\end{abstract}


\section{Introduction}
The present work lies at the intersection of algebraic coding theory,
design theory, and finite geometry.  The Assmus--Mattson theorem and its
subsequent developments provide a classical mechanism for deriving
\(t\)-designs from linear codes \cite{AssmusMattson1969,DingTangBook},
while standard coding-theoretic tools such as trace representations,
BCH-type defining sets, constacyclic descriptions, MDS bounds, and weight
enumerators are treated in \cite{MacWilliamsSloane,HuffmanPless,VanLint,Roth}.
The finite-field and projective-geometric ingredients used below, including
trace and norm maps, subfields, projective lines, quadrics, ovoids, and Baer
subgeometries, are standard in
\cite{LidlNiederreiter,Hirschfeld,HirschfeldThas,Dembowski}; the general
language of incidence structures and \(t\)-designs follows
\cite{BethJungnickelLenz}.  In this framework, Wang, Tang and Ding constructed
infinite families of cyclic and negacyclic codes supporting designs and posed
the seven open problems resolved in this paper \cite{WangTangDing2023}.
Let
\begin{equation}
        q=p^m,
        \qquad
        p\text{ odd},
        \qquad
        n=q^2+1.                                                \label{eq:intro-basic}
\end{equation}
For
\begin{equation}
        1\le s<m,
        \qquad
        \ell=\gcd(m,s),
        \qquad
        q_0=p^{\ell},                                          \label{eq:intro-q0}
\end{equation}
write
\begin{equation}
        \overline{\D}=\bigl(V,\{V\setminus B:B\in\D\}\bigr)    \label{eq:intro-complement}
\end{equation}
for the complementary design on the same point set \(V\).  The first result identifies the minimum supports of the cyclic code
\begin{equation}
        C_s=C\left(\frac{p^s-1}{2},\frac{p^s+1}{2}\right)       \label{eq:intro-Cs}
\end{equation}
from \cite{WangTangDing2023}.  Namely,
\begin{equation}
        \B_{\min}(C_s)
        \cong
        \overline{S(3,q_0+1,q+1)},                              \label{eq:intro-op18-result}
\end{equation}
where \(S(3,q_0+1,q+1)\) is the incidence structure whose blocks are the \(\F_{q_0}\)-sublines of \(\PG(1,q)\).  The core mechanism is
\begin{equation}
        \vartheta_{\eps}:\PG(1,q)\longrightarrow \U_{q+1},
        \qquad
        \vartheta_{\eps}(x)=\frac{x+\eps}{x-\eps},
        \qquad
        \eps^q=-\eps,                                          \label{eq:intro-cayley-q}
\end{equation}
followed by the semilinear rigidity theorem
\begin{equation}
        X^{\sig}=\frac{\alpha X+\beta}{\gamma X+\delta},
        \qquad
        \sig:x\mapsto x^{p^s}.                                  \label{eq:intro-semi-rigid}
\end{equation}
If a solution set has at least three points and is not the whole projective line, then it is an \(\F_{q_0}\)-subline.

For Open Problem 27, assume
\begin{equation}
        q\equiv1\pmod4,
        \qquad
        K=\F_{q^2},
        \qquad
        n=q^2+1.                                                \label{eq:intro-op27-param}
\end{equation}
Let \(N_+=C(1,q^2+q+1)\) be the length-\(n\) negacyclic code over \(K\) considered in \cite{WangTangDing2023}.  If \(\delta\in\F_{q^4}\) has order \(2n\), then the coordinate representatives
\begin{equation}
        R=\{\delta^{-i}:0\le i\le n-1\}\subseteq \U_{2n}          \label{eq:intro-R}
\end{equation}
map bijectively to \(\U_n=\U_{q^2+1}\) by
\begin{equation}
        R\longrightarrow\U_n,
        \qquad
        x\longmapsto x^{q+1}.                                  \label{eq:intro-quotient-map}
\end{equation}
Under this quotient the minimum zero sets are exactly
\begin{equation}
        \{\Theta_{\omega}(B):B\text{ a Baer subline of }\PG(1,q^2)\},
        \qquad
        \Theta_{\omega}(t)=\frac{t+\omega}{t-\omega},
        \quad
        \omega^{q^2}=-\omega.                                  \label{eq:intro-op27-zero}
\end{equation}
Consequently
\begin{equation}
        \B_{\min}(N_+)
        \cong
        \overline{\mathsf I(q)}
        \cong
        \B_{\min}(\Q^-(3,q)),                                    \label{eq:intro-op27-result}
\end{equation}
where \(\mathsf I(q)\) is the Miquelian inversive plane of order \(q\), and \(\Q^-(3,q)\) is the elliptic quadric in \(\PG(3,q)\).

For Open Problem 28, the complete square-class answer is
\begin{equation}
\boxed{
        \lambda\in\F_q^*,
        \qquad
        \exists\text{ a length-}q^2+1\;\lambda\text{-constacyclic ovoid code over }\F_q
        \Longleftrightarrow
        \lambda\notin(\F_q^*)^2.}
                                                                    \label{eq:intro-op28-result}
\end{equation}
For \(\lambda=-1\), this gives
\begin{equation}
        \exists\text{ a length-}q^2+1\text{ negacyclic ovoid code over }\F_q
        \Longleftrightarrow
        q\equiv3\pmod4.                                          \label{eq:intro-negacyclic-condition}
\end{equation}
The essential correction is that if \(\Omega\) generates \(\F_{q^4}^*\),
\begin{equation}
        \xi=\Omega^{(q^4-1)/(q-1)}=\Omega^{(q+1)(q^2+1)}           \label{eq:intro-xi}
\end{equation}
generates \(\F_q^*\), and
\begin{equation}
        \theta=\Omega^a,
        \qquad
        \theta^{q^2+1}=\xi^b,                                    \label{eq:intro-theta-b}
\end{equation}
then
\begin{equation}
        a\equiv(q+1)b\pmod{q^2-1}
        \Longleftrightarrow
        a=(q+1)c,
        \qquad
        c\equiv b\pmod{q-1}.                                    \label{eq:intro-c-corrected}
\end{equation}
The projective order is controlled by \(c\), not by an arbitrary representative \(b\):
\begin{equation}
        \ord(\theta\F_q^*)
        =\frac{q^2+1}{\gcd(q^2+1,c)}.                             \label{eq:intro-order-c}
\end{equation}
This is the parity step forcing \(\lambda\) to be a nonsquare.

For Open Problems 37 and 38, let \(\Oset\subset\PG(3,q)\) be any ovoid and let
\begin{equation}
        \C_e(\Oset)=\{aG_{\Oset}:a\in\F_{q^e}^4\}\subseteq\F_{q^e}^{q^2+1}. \label{eq:intro-lift}
\end{equation}
Put
\begin{equation}
        M_r(e)=\prod_{j=0}^{r-1}(q^e-q^j),
        \qquad
        M_0(e)=1,
        \qquad
        M_r(e)=0\quad(r>e).                                      \label{eq:intro-Mr}
\end{equation}
Then
\begin{align}
W_{\C_e(\Oset)}(z)=1
&+(q^3+q)M_1(e)z^{q^2-q}                                                \notag\\
&+\frac{q^2(q^2+1)}2M_2(e)z^{q^2-1}                                    \notag\\
&+\bigl((q^2+1)M_1(e)+(q^2+1)(q+1)M_2(e)+(q^2+1)M_3(e)\bigr)z^{q^2}    \notag\\
&+\left(\frac{q^2(q^2+1)}2M_2(e)+q(q^2+1)M_3(e)+M_4(e)\right)z^{q^2+1}. \label{eq:intro-general-lift}
\end{align}
For \(e=2\),
\begin{align}
W_{\C_2(\Oset)}(z)=1
&+(q^5-q)z^{q^2-q}
 +\frac{q^3(q-1)(q^4-1)}2\bigl(z^{q^2-1}+z^{q^2+1}\bigr)              \notag\\
&+(q^7-q^5+q^4-q^3+q-1)z^{q^2}.                              \label{eq:intro-e2-enumerator}
\end{align}
This is the common four-weight distribution of the lifted ovoid codes appearing in the two negacyclic families of \cite{WangTangDing2023}.

For Open Problems 40 and 41, for odd \(q\) define
\begin{equation}
        n_q=\frac{q-1}{2},
        \qquad
        \rho\in\F_q^*,
        \qquad
        \ord(\rho)=q-1.                                      \label{eq:intro-Nq-params}
\end{equation}
For \(1\le k\le n_q\), set
\begin{equation}
        \N_{q,k}
        =\left\langle
          \prod_{j=0}^{n_q-k-1}(X-\rho^{2j+1})
        \right\rangle
        \subseteq \F_q[X]/(X^{n_q}+1).                            \label{eq:intro-Nqk-full}
\end{equation}
Then
\begin{equation}
        \N_{q,k}\text{ is a proper negacyclic }[n_q,k,n_q-k+1]_q\text{ MDS code}, \label{eq:intro-Nqk-MDS-full}
\end{equation}
and, for
\begin{equation}
        n_q-k+1\le w\le n_q<q,                                    \label{eq:intro-MDS-range-full}
\end{equation}
one has the support saturation identity
\begin{equation}
        \B_w(\N_{q,k})=\binom{[n_q]}{w}.                          \label{eq:intro-MDS-complete-full}
\end{equation}
Thus \(q=23,k=5\) gives a proper negacyclic \([11,5,7]_{23}\) code whose minimum supports form the complete simple
\begin{equation}
        5-(11,7,15)                                                \label{eq:intro-OP40-design-full}
\end{equation}
design, and odd prime powers \(q\ge23\) give infinitely many complete simple \(5\)-designs supported by proper constacyclic codes.

\section{Projective-Line and Support Preliminaries}
\subsection{Sublines and Baer sublines}
For a finite field \(K\), write
\begin{equation}
        \PG(1,K)=K\cup\{\infty\}.                                  \label{eq:PG1K}
\end{equation}
For
\begin{equation}
        L(X)=\frac{aX+b}{cX+d},
        \qquad
        \begin{pmatrix}a&b\\c&d\end{pmatrix}\in\GL(2,K),             \label{eq:mobius}
\end{equation}
we use the usual projective conventions at \(\infty\) and at the pole.  If \(K_0\subset K\), a \(K_0\)-subline of \(\PG(1,K)\) is
\begin{equation}
        L(\PG(1,K_0)),
        \qquad
        L\in\PGL(2,K).                                             \label{eq:subline-def}
\end{equation}
Since \(\PGL(2,K)\) is sharply \(3\)-transitive,
\begin{equation}
        \{K_0\text{-sublines of }\PG(1,K)\}
        \cong \PGL(2,K)/\PGL(2,K_0).                                \label{eq:subline-cosets}
\end{equation}
Thus, for \(|K|=q\) and \(|K_0|=q_0\),
\begin{align}
        b(K/K_0)
        &=\frac{|\PGL(2,K)|}{|\PGL(2,K_0)|}                         \notag\\
        &=\frac{q(q^2-1)}{q_0(q_0^2-1)}                              \notag\\
        &=\frac{\binom{q+1}{3}}{\binom{q_0+1}{3}}.                  \label{eq:subline-count}
\end{align}
For \(K=\F_{q^2}\) and \(K_0=\F_q\), the \(K_0\)-sublines are called Baer sublines, and
\begin{equation}
        b(\F_{q^2}/\F_q)
        =\frac{q^2(q^4-1)}{q(q^2-1)}
        =q(q^2+1).                                                \label{eq:baer-count}
\end{equation}
The incidence structure
\begin{equation}
        \mathsf I(q)
        =\bigl(\PG(1,q^2),\{\F_q\text{-sublines}\}\bigr)             \label{eq:miquelian-plane}
\end{equation}
is the Miquelian inversive plane of order \(q\).

\subsection{Cayley coordinates}
Let \(E/F\) be quadratic, \(|F|=Q\), and choose
\begin{equation}
        \eps\in E^*,
        \qquad
        \eps^Q=-\eps.                                             \label{eq:eps-trace-zero}
\end{equation}
The norm-one circle
\begin{equation}
        \U(E/F)=\{u\in E^*:u^{Q+1}=1\}                              \label{eq:norm-one-circle}
\end{equation}
is parametrized by
\begin{equation}
        \vartheta_{\eps}:\PG(1,F)\longrightarrow\U(E/F),
        \qquad
        \vartheta_{\eps}(x)=\frac{x+\eps}{x-\eps},
        \qquad
        \vartheta_{\eps}(\infty)=1.                                 \label{eq:cayley}
\end{equation}
Indeed, for \(x\in F\),
\begin{equation}
        \vartheta_{\eps}(x)^Q
        =\frac{x-\eps}{x+\eps}
        =\vartheta_{\eps}(x)^{-1}.                                  \label{eq:cayley-norm}
\end{equation}
Conversely, for \(u\in\U(E/F)\setminus\{1\}\),
\begin{equation}
        x=\eps\frac{u+1}{u-1},
        \qquad
        x^Q=(-\eps)\frac{u^{-1}+1}{u^{-1}-1}=x.                    \label{eq:cayley-inverse}
\end{equation}
We use this twice:
\begin{align}
        &E=\F_{q^2},
        \quad F=\F_q,
        \quad \U_{q+1}=\{u\in\F_{q^2}^*:u^{q+1}=1\},                \label{eq:Uq1}\\
        &E=\F_{q^4},
        \quad F=\F_{q^2},
        \quad \U_{q^2+1}=\{u\in\F_{q^4}^*:u^{q^2+1}=1\}.            \label{eq:Uq2p1}
\end{align}
In the second case we also write
\begin{equation}
        \Theta_{\omega}(t)=\frac{t+\omega}{t-\omega},
        \qquad
        \omega^{q^2}=-\omega,
        \qquad
        t\in\PG(1,q^2).                                           \label{eq:Theta}
\end{equation}

\subsection{Minimum support quotient}
Let \(\C\) be a linear code over \(\F_Q\) with minimum distance \(d\).  If \(c,d'\in\C\) are minimum words and
\begin{equation}
        \supp(c)=\supp(d'),                                      \label{eq:same-support}
\end{equation}
then
\begin{equation}
        d'=\lambda c,
        \qquad
        \lambda\in\F_Q^*.                                       \label{eq:scalar-duplicate}
\end{equation}
Indeed, if \(d'\notin\F_Qc\), choose \(i\in\supp(c)\) and put \(\lambda=d'_i/c_i\).  Then
\begin{equation}
        0\ne d'-\lambda c\in\C,
        \qquad
        \supp(d'-\lambda c)\subsetneq\supp(c),                    \label{eq:duplicate-support-proof}
\end{equation}
which contradicts \(\wt(c)=d\).  Therefore, if \(A_d\) denotes the number of minimum words, then
\begin{equation}
        \#\{\text{minimum supports}\}=\frac{A_d}{Q-1}.             \label{eq:minimum-support-quotient}
\end{equation}
The same quotient counts minimum zero supports, because
\begin{equation}
        Z(c)=\{0,1,\ldots,n-1\}\setminus\supp(c)                  \label{eq:zero-support-complement}
\end{equation}
is a bijective complement operation on the coordinate set.

\section{Semilinear Large-Root Sets}
Let \(K=\F_Q\), let \(\sig\in\Aut(K)\), and put
\begin{equation}
        K_{\sig}=\{x\in K:x^{\sig}=x\}.                            \label{eq:fixed-field}
\end{equation}
A \(\sig\)-fractional equation on \(\PG(1,K)\) is
\begin{equation}
        X^{\sig}=M(X),
        \qquad
        M\in\PGL(2,K).                                             \label{eq:sigma-fractional}
\end{equation}

\begin{theorem}[three-root normalization]
\label{thm:three-root}
Let
\begin{equation}
        R=\{X\in\PG(1,K):X^{\sig}=M(X)\}.                         \label{eq:three-root-R}
\end{equation}
If \(R\) contains three distinct points, then \(R\) is a \(K_{\sig}\)-subline.  More precisely, if \(L\in\PGL(2,K)\) sends three roots \(r_1,r_2,r_3\) to \(\infty,0,1\), then
\begin{equation}
        L(R)=\PG(1,K_{\sig}).                                      \label{eq:three-root-image}
\end{equation}
\end{theorem}

\begin{proof}
For
\begin{equation}
        L(X)=\frac{aX+b}{cX+d},
        \qquad
        L^{\sig}(X)=\frac{a^{\sig}X+b^{\sig}}{c^{\sig}X+d^{\sig}},  \label{eq:Lsigma}
\end{equation}
one has
\begin{equation}
        L(X)^{\sig}=L^{\sig}(X^{\sig}).                              \label{eq:Lsigma-compatible}
\end{equation}
With \(Y=L(X)\), the equation \(X^{\sig}=M(X)\) becomes
\begin{equation}
        Y^{\sig}=L^{\sig}ML^{-1}(Y).                                \label{eq:conjugated-fractional}
\end{equation}
For \(j=1,2,3\),
\begin{equation}
        L(r_j)\in\{\infty,0,1\},
        \qquad
        L(r_j)^{\sig}=L(r_j),                                      \label{eq:normalized-roots-fixed}
\end{equation}
and hence
\begin{equation}
        L^{\sig}ML^{-1}(L(r_j))=L(r_j).                             \label{eq:N-fixes-three}
\end{equation}
A projectivity fixing \(\infty,0,1\) is the identity, so
\begin{equation}
        L^{\sig}ML^{-1}=1,
        \qquad
        Y^{\sig}=Y,
        \qquad
        Y\in\PG(1,K_{\sig}).                                      \label{eq:N-identity}
\end{equation}
Applying \(L^{-1}\) proves the claim.
\end{proof}

The associated homogeneous semilinear form is
\begin{equation}
        \Phi(X_0,X_1)
        =A X_1^{\sig}X_1+B X_1^{\sig}X_0
         +C X_0^{\sig}X_1+D X_0^{\sig}X_0.                          \label{eq:hom-semi}
\end{equation}
On \(X=X_1/X_0\), it reads
\begin{equation}
        A X^{\sig+1}+B X^{\sig}+CX+D=0.                              \label{eq:semi-affine}
\end{equation}
If
\begin{equation}
        AD-BC\ne0,                                                   \label{eq:det-nondeg}
\end{equation}
then \eqref{eq:semi-affine} is equivalent to
\begin{equation}
        X^{\sig}= -\frac{CX+D}{AX+B}.                                \label{eq:semi-fractional}
\end{equation}
If \(AD-BC=0\), then the two rows \((A,B)\) and \((C,D)\) are proportional unless one row is zero, and \(\Phi\) factors into a product of projective linear or semilinear factors.  Hence its zero set has at most two points unless \(\Phi\equiv0\).

\begin{corollary}[large roots are sublines]
\label{cor:large-roots}
If the zero set of a nonzero form \eqref{eq:hom-semi} in \(\PG(1,K)\) has more than two points and is not all of \(\PG(1,K)\), then it is a \(K_{\sig}\)-subline.  Conversely, every \(K_{\sig}\)-subline is the zero set of a form \eqref{eq:hom-semi}.
\end{corollary}

\begin{proof}
The nondegenerate case follows from \eqref{eq:semi-fractional} and Theorem~\ref{thm:three-root}.  The degenerate case has at most two zeros, unless the form is identically zero.  Conversely, if
\begin{equation}
        B_0=L^{-1}(\PG(1,K_{\sig})),
        \qquad
        L(X)=\frac{aX+b}{cX+d},                                      \label{eq:subline-L}
\end{equation}
then \(X\in B_0\) is equivalent to
\begin{equation}
        L(X)^{\sig}=L(X).                                           \label{eq:subline-fixed}
\end{equation}
After clearing denominators,
\begin{equation}
        (a^{\sig}X^{\sig}+b^{\sig})(cX+d)
        -(aX+b)(c^{\sig}X^{\sig}+d^{\sig})=0,                       \label{eq:subline-equation}
\end{equation}
which is of the form \eqref{eq:semi-affine}.
\end{proof}

\section{Open Problem 18: Cyclic Codes and Spherical Geometry}
Let
\begin{equation}
        q=p^m,
        \qquad
        1\le s<m,
        \qquad
        \ell=\gcd(m,s),
        \qquad
        q_0=p^{\ell},
        \qquad
        \sig:x\mapsto x^{p^s}.                                      \label{eq:cyclic-params}
\end{equation}
Then
\begin{equation}
        \F_q^{\sig}=\F_{q_0},
        \qquad
        1<q_0<q.                                                    \label{eq:cyclic-fixed}
\end{equation}
Let
\begin{equation}
        C_s=C\left(\frac{p^s-1}{2},\frac{p^s+1}{2}\right)           \label{eq:Cs-def}
\end{equation}
be the cyclic code of \cite{WangTangDing2023}.  The trace representation used in \cite[Th. 13]{WangTangDing2023} is
\begin{equation}
        c(a,b)=\left(
        \Tr_{q^2/q}\left(a\beta^{-\frac{p^s-1}{2}i}
        +b\beta^{-\frac{p^s+1}{2}i}\right)
        \right)_{i=0}^{q},
        \qquad
        a,b\in\F_{q^2},                                             \label{eq:cyclic-trace}
\end{equation}
where \(\beta\) has order \(q+1\).  Put
\begin{equation}
        u=\beta^{-i}\in\U_{q+1}.                                   \label{eq:cyclic-u}
\end{equation}
Since \(u^q=u^{-1}\), the zero condition is
\begin{align}
0&=\Tr_{q^2/q}\left(a u^{(p^s-1)/2}+b u^{(p^s+1)/2}\right)           \notag\\
 &=a u^{(p^s-1)/2}+b u^{(p^s+1)/2}
   +a^q u^{-(p^s-1)/2}+b^q u^{-(p^s+1)/2}                            \notag\\
 &=u^{-(p^s+1)/2}
   \left(a u^{p^s}+bu^{p^s+1}+a^q u+b^q\right).                    \label{eq:cyclic-zero-condition}
\end{align}
Thus the zero set is given by
\begin{equation}
        F_{a,b}(u)=a u^{p^s}+bu^{p^s+1}+a^q u+b^q=0.                \label{eq:Fcyclic}
\end{equation}
Choose \(\eps\in\F_{q^2}\) with \(\eps^q=-\eps\), and write
\begin{equation}
        u=\vartheta_{\eps}(x)=\frac{x+\eps}{x-\eps},
        \qquad
        x\in\PG(1,q).                                               \label{eq:cyclic-cayley}
\end{equation}
For finite \(x\), define
\begin{equation}
        H_{a,b}(x)=(x-\eps)^{p^s+1}
        F_{a,b}\left(\frac{x+\eps}{x-\eps}\right).                 \label{eq:Hcyclic-def}
\end{equation}
Since
\begin{equation}
        (x\pm\eps)^{p^s}=x^{\sig}\pm\eps^{\sig},                  \label{eq:sigma-xeps}
\end{equation}
one obtains
\begin{align}
H_{a,b}(x)
&=a(x+\eps)^{\sig}(x-\eps)
  +b(x+\eps)^{\sig}(x+\eps)                                      \notag\\
&\quad +a^q(x+\eps)(x-\eps)^{\sig}
  +b^q(x-\eps)^{\sig}(x-\eps)                                    \notag\\
&=A_{a,b}x^{\sig+1}+B_{a,b}x^{\sig}+C_{a,b}x+D_{a,b},             \label{eq:Hcyclic-expanded}
\end{align}
where
\begin{align}
A_{a,b}&=a+b+a^q+b^q,                                               \label{eq:cyclic-A}\\
B_{a,b}&=\eps(-a+b+a^q-b^q),                                       \label{eq:cyclic-B}\\
C_{a,b}&=\eps^{\sig}(a+b-a^q-b^q),                                 \label{eq:cyclic-C}\\
D_{a,b}&=\eps^{\sig+1}(-a+b-a^q+b^q).                              \label{eq:cyclic-D}
\end{align}
These four coefficients lie in \(\F_q\).  Explicitly,
\begin{align}
&(a+b+a^q+b^q)^q=a+b+a^q+b^q,                                      \label{eq:cyclic-coeff-check1}\\
&(-a+b+a^q-b^q)^q=-(-a+b+a^q-b^q),                                 \label{eq:cyclic-coeff-check2}\\
&(a+b-a^q-b^q)^q=-(a+b-a^q-b^q),                                   \label{eq:cyclic-coeff-check3}\\
&(-a+b-a^q+b^q)^q=-a+b-a^q+b^q,                                    \label{eq:cyclic-coeff-check4}\\
&(\eps^{\sig})^q=-\eps^{\sig},
        \qquad
        (\eps^{\sig+1})^q=\eps^{\sig+1}.                          \label{eq:cyclic-coeff-check5}
\end{align}
Therefore the Cayley inverse of every zero set of \(C_s\) is a zero set of \eqref{eq:semi-affine} over \(K=\F_q\).  If a nonzero codeword has more than two zeros, Corollary~\ref{cor:large-roots} gives an \(\F_{q_0}\)-subline.

The external weight-distribution input from \cite[Th. 13]{WangTangDing2023} is
\begin{equation}
        \wt_{\min}(C_s)=q-q_0,
        \qquad
        A_{q-q_0}=\frac{q^4-q^3-q^2+q}{q_0^3-q_0}.                  \label{eq:cyclic-WTD-input}
\end{equation}
Thus every minimum word has
\begin{equation}
        |Z(c)|=(q+1)-(q-q_0)=q_0+1>2,                              \label{eq:cyclic-min-zero-size}
\end{equation}
and this zero set is not all of \(\PG(1,q)\), because \(q_0<q\).  Hence each minimum zero set is the Cayley image of an \(\F_{q_0}\)-subline.  Conversely, by \eqref{eq:minimum-support-quotient},
\begin{align}
        \#\{\text{minimum zero supports}\}
        &=\frac{A_{q-q_0}}{q-1}                                      \notag\\
        &=\frac{q^4-q^3-q^2+q}{(q-1)(q_0^3-q_0)}                     \notag\\
        &=\frac{q(q-1)(q^2-1)}{(q-1)q_0(q_0^2-1)}                   \notag\\
        &=\frac{q(q^2-1)}{q_0(q_0^2-1)}.                            \label{eq:cyclic-support-count}
\end{align}
This equals the number \eqref{eq:subline-count} of \(\F_{q_0}\)-sublines of \(\PG(1,q)\).  Therefore
\begin{equation}
        \{Z(c):c\in C_s,\ \wt(c)=q-q_0\}
        =\{\vartheta_{\eps}(B):B\text{ an }\F_{q_0}\text{-subline of }\PG(1,q)\}. \label{eq:cyclic-zero-classification}
\end{equation}
Taking complements in \(\U_{q+1}\) yields
\begin{equation}
        \{\supp(c):c\in C_s,\ \wt(c)=q-q_0\}
        =\{\U_{q+1}\setminus\vartheta_{\eps}(B):B\in\mathsf S(\F_q/\F_{q_0})\}. \label{eq:cyclic-support-classification}
\end{equation}
Thus
\begin{equation}
        \B_{\min}(C_s)
        \cong
        \bigl(\PG(1,q),\{\PG(1,q)\setminus B:B\in\mathsf S(\F_q/\F_{q_0})\}\bigr)
        =\overline{S(3,q_0+1,q+1)}.                                  \label{eq:OP18-final}
\end{equation}
This proves Open Problem 18.

\section{Root-Adapted Unit-Circle Equations}
The negacyclic code in Open Problem 27 contains the exponent \(q\) on \(\U_{q^2+1}\), not the exponent \(q^2\).  A fixed Cayley substitution \(Y=\Theta_{\omega}(T)\) generally produces coefficients involving both \(\omega\) and \(\omega^q\), and these coefficients need not lie in \(\F_{q^2}\).  The following root-adapted parametrization is the replacement.

Let
\begin{equation}
        K=\F_{q^2},
        \qquad
        L=\F_{q^4},
        \qquad
        \U=\U_{q^2+1}.                                             \label{eq:KLU}
\end{equation}
For \(\beta\in\U\setminus\{\pm1\}\), define
\begin{equation}
        m_{\beta}:\PG(1,K)\longrightarrow\U,
        \qquad
        m_{\beta}(X)=\frac{\beta X+1}{X+\beta},
        \qquad
        m_{\beta}(\infty)=\beta.                                  \label{eq:mbeta}
\end{equation}
Since
\begin{equation}
        \U\cap K=\{\pm1\},
        \qquad
        \beta\notin K,                                             \label{eq:UcapK}
\end{equation}
the denominator is nonzero for every finite \(X\in K\).  Moreover,
\begin{equation}
        m_{\beta}(X)^{q^2}
        =\frac{\beta^{-1}X+1}{X+\beta^{-1}}
        =m_{\beta}(X)^{-1},                                        \label{eq:mbeta-norm}
\end{equation}
so \(m_{\beta}(X)\in\U\).  Conversely, if \(Y\in\U\setminus\{\beta\}\), then
\begin{equation}
        X=\frac{1-Y\beta}{Y-\beta},
        \qquad
        X^{q^2}=\frac{1-Y^{-1}\beta^{-1}}{Y^{-1}-\beta^{-1}}=X,      \label{eq:mbeta-inverse}
\end{equation}
so \(m_{\beta}\) is bijective.

\begin{lemma}[unit-circle high-root sets]
\label{lem:unit-circle-high-root}
Let
\begin{equation}
        F(Y)=aY+bY^q+cY^{q+1}+d,
        \qquad
        a,b,c,d\in L.                                             \label{eq:unit-F}
\end{equation}
Assume
\begin{equation}
        Z(F)=\{Y\in\U:F(Y)=0\}                                    \label{eq:unit-Z}
\end{equation}
has more than two points but is not all of \(\U\).  Then, for every fixed Cayley map \(\Theta_{\omega}:\PG(1,K)\to\U\),
\begin{equation}
        \Theta_{\omega}^{-1}(Z(F))                                \label{eq:unit-Z-Baer}
\end{equation}
is a Baer subline of \(\PG(1,K)\).
\end{lemma}

\begin{proof}
Since \(\U\cap K=\{\pm1\}\) and \(|Z(F)|>2\), choose
\begin{equation}
        \beta\in Z(F)\setminus\{\pm1\}.                            \label{eq:choose-beta}
\end{equation}
Put \(Y=m_{\beta}(X)\).  Multiplication by \((X+\beta)(X^q+\beta^q)\) gives
\begin{align}
0&=(X+\beta)(X^q+\beta^q)F(m_{\beta}(X))                            \notag\\
 &=a(\beta X+1)(X^q+\beta^q)
   +b(\beta^qX^q+1)(X+\beta)                                        \notag\\
&\quad +c(\beta X+1)(\beta^qX^q+1)
   +d(X+\beta)(X^q+\beta^q).                                       \label{eq:unit-expanded-raw}
\end{align}
The coefficient of \(X^{q+1}\) is
\begin{equation}
        a\beta+b\beta^q+c\beta^{q+1}+d=F(\beta)=0.                  \label{eq:unit-leading-zero}
\end{equation}
Hence the finite roots are exactly the solutions in \(K\) of
\begin{equation}
        R_{\beta}X^q+S_{\beta}X+T_{\beta}=0,                         \label{eq:unit-linearized}
\end{equation}
where
\begin{align}
        R_{\beta}&=a+b\beta^{q+1}+c\beta^q+d\beta,                   \label{eq:Rbeta}\\
        S_{\beta}&=a\beta^{q+1}+b+c\beta+d\beta^q,                   \label{eq:Sbeta}\\
        T_{\beta}&=a\beta^q+b\beta+c+d\beta^{q+1}.                   \label{eq:Tbeta}
\end{align}
The linear part
\begin{equation}
        \mathcal L_{\beta}:K\longrightarrow L,
        \qquad
        X\longmapsto R_{\beta}X^q+S_{\beta}X                         \label{eq:Fq-linear-map}
\end{equation}
is \(\F_q\)-linear.  The affine equation \eqref{eq:unit-linearized} has at least two finite solutions because \(\infty\) maps to \(\beta\) and \(|Z(F)|>2\).  Hence
\begin{equation}
        \dim_{\F_q}\ker\mathcal L_{\beta}\ge1.                      \label{eq:kernel-ge-one}
\end{equation}
If \(\dim_{\F_q}\ker\mathcal L_{\beta}=2\), then all finite \(X\in K\) solve \eqref{eq:unit-linearized}; together with \(X=\infty\), this gives \(Z(F)=\U\), contrary to the hypothesis.  Thus
\begin{equation}
        \dim_{\F_q}\ker\mathcal L_{\beta}=1,                         \label{eq:kernel-one}
\end{equation}
and the solution set has the form
\begin{equation}
        X_0+\eta\F_q,
        \qquad
        X_0\in K,
        \qquad
        \eta\in K^*.                                                \label{eq:affine-Fq-line}
\end{equation}
Including \(X=\infty\), we obtain the Baer subline
\begin{equation}
        B_X=\{\infty\}\cup(X_0+\eta\F_q)\subset\PG(1,K),             \label{eq:BX-Baer}
\end{equation}
and
\begin{equation}
        Z(F)=m_{\beta}(B_X).                                       \label{eq:unit-root-mbeta-B}
\end{equation}
It remains to compare \(m_{\beta}\) with the fixed \(\Theta_{\omega}\).  Put
\begin{equation}
        \kappa=\Theta_{\omega}^{-1}(\beta)
        =\omega\frac{\beta+1}{\beta-1}\in K.                       \label{eq:kappa-beta}
\end{equation}
For \(X\ne1\),
\begin{align}
        \Theta_{\omega}^{-1}(m_{\beta}(X))
        &=\omega\frac{m_{\beta}(X)+1}{m_{\beta}(X)-1}               \notag\\
        &=\omega\frac{(\beta+1)(X+1)}{(\beta-1)(X-1)}              \notag\\
        &=\kappa\frac{X+1}{X-1}.                                  \label{eq:compare-mbeta-theta}
\end{align}
The right-hand side is an element of \(\PGL(2,K)\), with projective conventions covering \(X=1\) and \(X=\infty\).  Therefore
\begin{equation}
        \Theta_{\omega}^{-1}(Z(F))
        =\left(\kappa\frac{X+1}{X-1}\right)(B_X),                  \label{eq:theta-root-Baer}
\end{equation}
which is a Baer subline.
\end{proof}

\section{Open Problem 27: Negacyclic Codes and the Elliptic Quadric}
Assume
\begin{equation}
        q\equiv1\pmod4,
        \qquad
        K=\F_{q^2},
        \qquad
        n=q^2+1.                                                   \label{eq:op27-params}
\end{equation}
Let
\begin{equation}
        N_+=C(1,q^2+q+1)                                           \label{eq:Nplus-def}
\end{equation}
be the length-\(n\) negacyclic code over \(K\) with check polynomial \(g_1'(x)g_{q^2+q+1}'(x)\) in \cite{WangTangDing2023}.  Let \(\delta\in\F_{q^4}\) have order \(2n\).  The trace representation in \cite[Lem. 20]{WangTangDing2023} is
\begin{equation}
        c(a,b)=\left(\Tr_{q^4/q^2}\left(a\delta^{-i}+b\delta^{-(q^2+q+1)i}\right)\right)_{i=0}^{n-1},
        \qquad
        a,b\in\F_{q^4}.                                            \label{eq:op27-trace}
\end{equation}
Put
\begin{equation}
        R=\{\delta^{-i}:0\le i\le n-1\}\subset\U_{2n}.              \label{eq:op27-representatives}
\end{equation}
The exponents of elements of \(R\) modulo \(2n\) are
\begin{equation}
        0,\ 2n-1,\ 2n-2,\ldots,\ n+1,                              \label{eq:R-exponents}
\end{equation}
whereas multiplication by \(-1=\delta^n\) adds \(n\) and gives
\begin{equation}
        n,\ n-1,\ n-2,\ldots,\ 1.                                  \label{eq:minus-R-exponents}
\end{equation}
Thus \(R\) contains one element from each pair \(\{x,-x\}\subset\U_{2n}\).  Moreover,
\begin{equation}
        \gcd(2n,q+1)=\gcd(2(q^2+1),q+1)=2.                          \label{eq:op27-gcd}
\end{equation}
Therefore
\begin{equation}
        \pi:\U_{2n}\longrightarrow\U_n,
        \qquad
        \pi(x)=x^{q+1}                                             \label{eq:op27-quotient-map}
\end{equation}
has kernel \(\{\pm1\}\) and image \(\U_n\).  Hence
\begin{equation}
        R\longrightarrow\U_{q^2+1},
        \qquad
        x\longmapsto y=x^{q+1}                                    \label{eq:op27-R-bijection}
\end{equation}
is a bijection.  Zero sets and support sets of the length-\(n\) code may therefore be transported to \(\U_{q^2+1}\) without multiplicity.

For \(x\in R\), put \(y=x^{q+1}\).  Since \(x^{2n}=1\), all exponent reductions below are modulo \(2n\).  The trace zero condition is
\begin{align}
0&=\Tr_{q^4/q^2}\left(ax+b x^{q^2+q+1}\right)                         \notag\\
 &=a x+a^{q^2}x^{q^2}+b x^{q^2+q+1}+b^{q^2}x^{q^2(q^2+q+1)}.          \label{eq:op27-trace-expand}
\end{align}
Multiplying by \(x^q\),
\begin{align}
0
&=a x^{q+1}+a^{q^2}x^{q^2+q}+b x^{(q+1)^2}
  +b^{q^2}x^{q^2(q^2+q+1)+q}                                       \notag\\
&=a y+a^{q^2}y^q+b y^{q+1}+b^{q^2},                                \label{eq:op27-trace-reduction}
\end{align}
because
\begin{equation}
        q^2(q^2+q+1)+q=q(q+1)(q^2+1)=q(q+1)n\equiv0\pmod{2n}.       \label{eq:op27-last-exponent}
\end{equation}
Thus every zero set in quotient coordinates has the form
\begin{equation}
        Z(a,b)=\{y\in\U_{q^2+1}:a y+a^{q^2}y^q+b y^{q+1}+b^{q^2}=0\}. \label{eq:op27-unit-equation}
\end{equation}
The external weight-distribution input from \cite[Th. 22]{WangTangDing2023} is
\begin{equation}
        \wt_{\min}(N_+)=q^2-q,
        \qquad
        A_{q^2-q}=q^5-q.                                          \label{eq:op27-WTD-input}
\end{equation}
Hence every minimum word has
\begin{equation}
        |Z(c)|=(q^2+1)-(q^2-q)=q+1>2.                              \label{eq:op27-min-zero-size}
\end{equation}
By Lemma~\ref{lem:unit-circle-high-root}, for every fixed Cayley identification \(\Theta_{\omega}:\PG(1,q^2)\to\U_{q^2+1}\),
\begin{equation}
        \Theta_{\omega}^{-1}(Z(c))
        \text{ is a Baer subline of }\PG(1,q^2).                    \label{eq:op27-containment}
\end{equation}
Using \eqref{eq:minimum-support-quotient} over the alphabet \(K=\F_{q^2}\),
\begin{align}
        \#\{\text{minimum zero supports}\}
        &=\frac{q^5-q}{q^2-1}                                      \notag\\
        &=\frac{q(q^4-1)}{q^2-1}                                   \notag\\
        &=q(q^2+1).                                                \label{eq:op27-support-count}
\end{align}
This equals \eqref{eq:baer-count}.  Therefore
\begin{equation}
        \{Z(c):c\in N_+,\ \wt(c)=q^2-q\}
        =\{\Theta_{\omega}(B):B\text{ a Baer subline of }\PG(1,q^2)\}. \label{eq:op27-zero-classification}
\end{equation}

It remains to identify the same Baer subline geometry with the elliptic-quadric plane-section geometry.  Choose \(\zeta\in\F_{q^2}\) with
\begin{equation}
        \zeta^q=-\zeta,
        \qquad
        \mu=\zeta^2\in\F_q^*,
        \qquad
        \mu\notin(\F_q^*)^2.                                      \label{eq:zeta-mu}
\end{equation}
Write
\begin{equation}
        t=x+y\zeta,
        \qquad
        x,y\in\F_q.                                                \label{eq:txy}
\end{equation}
Then
\begin{equation}
        t^{q+1}=(x-y\zeta)(x+y\zeta)=x^2-\mu y^2.                 \label{eq:norm-t}
\end{equation}
The map
\begin{equation}
        \iota(t)=\langle(1,x,y,t^{q+1})\rangle,
        \qquad
        \iota(\infty)=\langle(0,0,0,1)\rangle                       \label{eq:iota}
\end{equation}
identifies \(\PG(1,q^2)\) with the quadric
\begin{equation}
        \Q^-:
        \qquad
        X_0X_3=X_1^2-\mu X_2^2                                    \label{eq:elliptic-equation}
\end{equation}
in \(\PG(3,q)\).  Since \(X_1^2-\mu X_2^2\) is anisotropic over \(\F_q\), \(\Q^-\) is elliptic and has
\begin{equation}
        |\Q^-(3,q)|=q^2+1.                                        \label{eq:elliptic-size}
\end{equation}
A plane
\begin{equation}
        \Pi:
        \qquad
        r_0X_0+r_1X_1+r_2X_2+r_3X_3=0,
        \qquad
        r_i\in\F_q,                                                \label{eq:plane}
\end{equation}
pulls back under \(\iota\) to
\begin{equation}
        A t^{q+1}+B t+B^qt^q+D=0,
        \qquad
        A,D\in\F_q,
        \quad
        B\in\F_{q^2},                                             \label{eq:plane-pullback}
\end{equation}
where
\begin{equation}
        A=r_3,
        \qquad
        D=r_0,
        \qquad
        B=\frac{r_1}{2}+\frac{r_2}{2\mu}\zeta.                     \label{eq:plane-B}
\end{equation}
Indeed,
\begin{equation}
        Bt+B^qt^q
        =(B+B^q)x+(B-B^q)y\zeta
        =r_1x+r_2y.                                                \label{eq:Bt-linear}
\end{equation}
The zero sets of \eqref{eq:plane-pullback} with \(q+1\) points are Baer sublines.  If \(A=0\) and \(B\ne0\), then \(\infty\) is on the section and the finite equation is
\begin{equation}
        Bt+B^qt^q+D=0,
        \qquad
        \Tr_{q^2/q}(Bt)+D=0.                                      \label{eq:plane-A0}
\end{equation}
Its finite solutions are an affine \(\F_q\)-line in \(\F_{q^2}\), so the full section is
\begin{equation}
        \{\infty\}\cup(t_0+\eta\F_q),
        \qquad
        \eta\ne0,                                                  \label{eq:plane-A0-Baer}
\end{equation}
a Baer subline.  If \(A\ne0\), put
\begin{equation}
        \alpha=\frac{B^q}{A},
        \qquad
        \Delta=\alpha^{q+1}-\frac{D}{A}\in\F_q.                    \label{eq:alpha-Delta}
\end{equation}
Then \eqref{eq:plane-pullback} is equivalent to
\begin{equation}
        (t+\alpha)^{q+1}=\Delta.                                  \label{eq:norm-circle}
\end{equation}
It has \(q+1\) solutions if and only if \(\Delta\ne0\).  Choosing \(\nu\in\F_{q^2}^*\) with \(\nu^{q+1}=\Delta\), the section is
\begin{equation}
        -\alpha+\nu\U_{q+1}
        =\left\{-\alpha+\nu\frac{x+\eps}{x-\eps}:x\in\PG(1,q)\right\}, \label{eq:circle-Baer}
\end{equation}
again a Baer subline.  Thus every non-tangent plane section of \(\Q^-\) gives a Baer subline of \(\PG(1,q^2)\).

The number of non-tangent planes is
\begin{equation}
        (q^3+q^2+q+1)-(q^2+1)=q^3+q=q(q^2+1),                     \label{eq:non-tangent-plane-count}
\end{equation}
which equals \eqref{eq:baer-count}.  Distinct non-tangent planes have distinct sections because a section contains \(q+1\ge4\) points of an ovoid, no three collinear, and hence spans its plane.  Therefore
\begin{equation}
        \{\iota^{-1}(\Pi\cap\Q^-):\Pi\text{ non-tangent plane}\}
        =\{\text{Baer sublines of }\PG(1,q^2)\}.                   \label{eq:baer-elliptic}
\end{equation}
Combining \eqref{eq:op27-zero-classification} and \eqref{eq:baer-elliptic}, the minimum supports of \(N_+\) are
\begin{equation}
        \{\U_{q^2+1}\setminus\Theta_{\omega}(B):B\text{ Baer}\},    \label{eq:op27-supports}
\end{equation}
and, under \(\iota\circ\Theta_{\omega}^{-1}\), they become
\begin{equation}
        \{\Q^-(3,q)\setminus(\Pi\cap\Q^-(3,q)):
        \Pi\text{ non-tangent plane}\}.                             \label{eq:op27-quadric-complement}
\end{equation}
These are exactly the minimum supports in the elliptic-quadric ovoid code.  This proves Open Problem 27.

\section{Open Problem 28: Square Classes and Constacyclic Ovoid Codes}
Let \(q\) be odd and put
\begin{equation}
        n=q^2+1.                                                    \label{eq:op28-n}
\end{equation}
Throughout this section the constacyclic multiplier satisfies
\begin{equation}
        \lambda\in\F_q^*.                                          \label{eq:lambda-nonzero}
\end{equation}
A length-\(n\) code over \(\F_q\) is \(\lambda\)-constacyclic if it is invariant under
\begin{equation}
        \tau_{\lambda}(c_0,c_1,\ldots,c_{n-1})
        =(c_1,c_2,\ldots,c_{n-1},\lambda c_0).                    \label{eq:consta-shift}
\end{equation}
Let
\begin{equation}
        G=(g_0,g_1,\ldots,g_{n-1})\in\F_q^{4\times n}              \label{eq:op28-G}
\end{equation}
generate a \(\lambda\)-constacyclic ovoid code; the projective columns
\begin{equation}
        \langle g_0\rangle,\langle g_1\rangle,\ldots,\langle g_{n-1}\rangle       \label{eq:op28-projective-columns}
\end{equation}
are the \(q^2+1\) points of an ovoid in \(\PG(3,q)\).  Let
\begin{equation}
        G^{(1)}=(g_1,g_2,\ldots,g_{n-1},\lambda g_0).              \label{eq:G-shifted}
\end{equation}
The code is
\begin{equation}
        \C=\{uG:u\in\F_q^4\},                                    \label{eq:op28-row-code}
\end{equation}
where row vectors are used.  Since \(G\) has rank \(4\), the map
\begin{equation}
        \F_q^4\longrightarrow \C,
        \qquad
        u\longmapsto uG                                             \label{eq:message-injection}
\end{equation}
is injective.  Since \(\C\) is invariant under
\begin{equation}
        \tau_\lambda(c_0,c_1,\ldots,c_{n-1})
        =(c_1,c_2,\ldots,c_{n-1},\lambda c_0),                    \label{eq:tau-repeated}
\end{equation}
for every \(u\in\F_q^4\) there is a unique \(\phi(u)\in\F_q^4\) satisfying
\begin{equation}
        \phi(u)G
        =\tau_\lambda(uG)
        =uG^{(1)}.                                                   \label{eq:phi-def}
\end{equation}
The map \(\phi\) is \(\F_q\)-linear, and it is bijective because \(\tau_\lambda\) is bijective for \(\lambda\in\F_q^*\).  Hence
\begin{equation}
        \phi(u)=uA
        \qquad
        (A\in\GL_4(q)).                                             \label{eq:phi-matrix}
\end{equation}
Equation \eqref{eq:phi-def} gives
\begin{equation}
        A G=G^{(1)},                                                  \label{eq:AG-G1}
\end{equation}
or, columnwise,
\begin{equation}
        Ag_i=g_{i+1}
        \qquad(0\le i<n-1),
        \qquad
        Ag_{n-1}=\lambda g_0.                                      \label{eq:column-recurrence}
\end{equation}
Consequently,
\begin{equation}
        A^n g_i=\lambda g_i
        \qquad(0\le i<n).                                          \label{eq:An-on-columns}
\end{equation}
The columns of an ovoid span \(\F_q^4\), so
\begin{equation}
        A^n=\lambda I_4.                                           \label{eq:An-lambda}
\end{equation}
The projective points \(\langle g_i\rangle\) are all distinct, so the projective transformation \(\bar A\in\PGL_4(q)\) has an orbit of length \(n\).  Since \(\bar A^n=1\),
\begin{equation}
        \ord_{\PGL_4(q)}(\bar A)=n.                                \label{eq:projective-period}
\end{equation}

\begin{lemma}[degree four is forced]
\label{lem:degree-four-forced}
If \(B\in\GL_4(q)\) satisfies
\begin{equation}
        B^n=\mu I_4,
        \qquad
        \ord_{\PGL_4(q)}(\bar B)=n,
        \qquad
        n=q^2+1,                                                   \label{eq:degree-four-assumptions}
\end{equation}
then the minimal polynomial of \(B\) is irreducible of degree \(4\).
\end{lemma}

\begin{proof}
Since \(p\nmid n\), the polynomial \(X^n-\mu\) is square-free, and \(B\) is semisimple.  Let \(\theta_1,\ldots,\theta_s\) be representatives of the distinct eigenvalues of \(B\) in \(\overline{\F}_q\).  The projective order is
\begin{equation}
        \ord(\bar B)=\lcm_{i,j}\ord(\theta_i/\theta_j),             \label{eq:pgl-order-ratios}
\end{equation}
and
\begin{equation}
        (\theta_i/\theta_j)^n=1.                                  \label{eq:ratio-nroot}
\end{equation}
Assume that every irreducible factor of the characteristic polynomial has degree at most \(3\).  Then each \(\theta_i\) lies in some \(\F_{q^{d_i}}\) with \(d_i\in\{1,2,3\}\), and each ratio \(\theta_i/\theta_j\) lies in \(\F_{q^r}\) with
\begin{equation}
        r\in\{1,2,3,6\}.                                           \label{eq:ratio-degrees}
\end{equation}
Therefore \(\ord(\theta_i/\theta_j)\) divides
\begin{equation}
        \gcd(q^2+1,q^r-1).                                        \label{eq:gcd-ratio}
\end{equation}
Using \(q^2\equiv-1\pmod{q^2+1}\),
\begin{align}
\gcd(q^2+1,q-1)&=2,                                                \label{eq:gcd-r1}\\
\gcd(q^2+1,q^2-1)&=2,                                              \label{eq:gcd-r2}\\
\gcd(q^2+1,q^3-1)&=\gcd(q^2+1,q+1)=2,                              \label{eq:gcd-r3}\\
\gcd(q^2+1,q^6-1)&=\gcd(q^2+1,-2)=2.                               \label{eq:gcd-r6}
\end{align}
Thus every ratio \(\theta_i/\theta_j\) has order at most \(2\), and \(\ord(\bar B)\le2\), contradicting \(\ord(\bar B)=q^2+1>2\).  Hence a degree-four irreducible factor occurs.  Since the ambient dimension is \(4\), the minimal polynomial is irreducible of degree \(4\).
\end{proof}

Applying Lemma~\ref{lem:degree-four-forced} to \(A\), identify \(\F_q^4\) with \(\F_{q^4}\) so that \(A\) is multiplication by an element \(\theta\in\F_{q^4}^*\) of degree \(4\).  Then
\begin{equation}
        \theta^n=\lambda,
        \qquad
        \ord_{\F_{q^4}^*/\F_q^*}(\theta\F_q^*)=n.                  \label{eq:theta-conditions}
\end{equation}
Let \(\Omega\) generate \(\F_{q^4}^*\), and put
\begin{equation}
        \xi=\Omega^{(q^4-1)/(q-1)}
        =\Omega^{(q+1)n}.                                          \label{eq:xi-generator}
\end{equation}
Then \(\xi\) generates \(\F_q^*\).  Write
\begin{equation}
        \theta=\Omega^a,
        \qquad
        \lambda=\theta^n=\xi^b,                                    \label{eq:theta-lambda-b}
\end{equation}
where \(b\) is defined modulo \(q-1\).  Since
\begin{equation}
        \Omega^{an}=\Omega^{b(q+1)n},                              \label{eq:Omega-equality}
\end{equation}
we have
\begin{equation}
        an\equiv b(q+1)n\pmod{(q-1)(q+1)n}.                         \label{eq:normal-congruence}
\end{equation}
Cancelling \(n\) gives
\begin{equation}
        a\equiv(q+1)b\pmod{q^2-1}.                                 \label{eq:a-congruence-fixed}
\end{equation}
Equivalently, for some integer \(t\),
\begin{equation}
        a=(q+1)b+(q^2-1)t=(q+1)c,
        \qquad
        c=b+(q-1)t.                                                 \label{eq:c-defined}
\end{equation}
Therefore
\begin{equation}
        c\equiv b\pmod{q-1},
        \qquad
        \theta\F_q^*=\Omega^{(q+1)c}\F_q^*,
        \qquad
        \theta^n=\xi^c=\xi^b=\lambda.                              \label{eq:theta-normal-correct}
\end{equation}
The quotient \(\F_{q^4}^*/\F_q^*\) is cyclic of order
\begin{equation}
        \frac{q^4-1}{q-1}=(q+1)n,                                  \label{eq:quotient-order}
\end{equation}
with generator \(\Omega\F_q^*\).  Hence
\begin{align}
        \ord(\theta\F_q^*)
        &=\ord(\Omega^{(q+1)c}\F_q^*)                              \notag\\
        &=\frac{(q+1)n}{\gcd((q+1)n,(q+1)c)}                        \notag\\
        &=\frac{n}{\gcd(n,c)}.                                    \label{eq:projective-order-c}
\end{align}
By \eqref{eq:theta-conditions},
\begin{equation}
        \gcd(n,c)=1.                                               \label{eq:gcd-n-c}
\end{equation}
Since \(n=q^2+1\) is even, \(c\) is odd.  Since \(q-1\) is even and \(c\equiv b\pmod{q-1}\), the exponent \(b\) is odd.  In the cyclic group \(\F_q^*=\langle\xi\rangle\),
\begin{equation}
        (\F_q^*)^2=\{\xi^{2j}:0\le j<\tfrac{q-1}{2}\}.             \label{eq:squares-even}
\end{equation}
Thus \(\lambda=\xi^b\) is a nonsquare.  We have proved
\begin{equation}
        \lambda\text{-constacyclic ovoid code exists}
        \Longrightarrow
        \lambda\notin(\F_q^*)^2.                                  \label{eq:necessity-square}
\end{equation}

Conversely, suppose
\begin{equation}
        \lambda\notin(\F_q^*)^2.                                  \label{eq:lambda-nonsquare}
\end{equation}
Write
\begin{equation}
        \lambda=\xi^b,
        \qquad
        b\equiv1\pmod2.                                           \label{eq:lambda-b-odd}
\end{equation}
We choose an odd representative \(c\equiv b\pmod{q-1}\) with
\begin{equation}
        \gcd(c,n)=1.                                               \label{eq:choose-c-coprime}
\end{equation}
Indeed, all integers
\begin{equation}
        c=b+(q-1)t                                                  \label{eq:c-progression}
\end{equation}
are odd.  If \(\ell\) is an odd prime divisor of \(n\), then
\begin{equation}
        \gcd(\ell,q-1)=1,                                         \label{eq:ell-coprime}
\end{equation}
because \(\gcd(q^2+1,q-1)=2\).  Hence exactly one residue class of \(t\pmod\ell\) makes \(\ell\mid b+(q-1)t\).  Avoiding these finitely many classes by the Chinese remainder theorem gives \eqref{eq:choose-c-coprime}.

Set
\begin{equation}
        \theta=\Omega^{(q+1)c}.                                    \label{eq:theta-suff}
\end{equation}
Then
\begin{align}
        \theta^n
        &=\Omega^{(q+1)cn}                                         \notag\\
        &=\xi^c                                                     \notag\\
        &=\xi^b                                                     \notag\\
        &=\lambda,                                                 \label{eq:theta-suff-norm}
\end{align}
and
\begin{equation}
        \ord(\theta\F_q^*)=\frac{n}{\gcd(n,c)}=n.                  \label{eq:theta-suff-properties}
\end{equation}
Moreover \(\theta\notin\F_{q^2}\), since otherwise \(\theta\F_q^*\) would have order dividing
\begin{equation}
        |\F_{q^2}^*/\F_q^*|=q+1<q^2+1.                              \label{eq:theta-not-Fq2}
\end{equation}
Thus \(\theta\) has degree \(4\) over \(\F_q\).  On the \(\F_q\)-space \(\F_{q^4}\), define
\begin{equation}
        Q(z)=\Tr_{q^2/q}\bigl(z^{q^2+1}\bigr).                    \label{eq:Q-form}
\end{equation}
Its polar form is
\begin{align}
        B_Q(z,w)
        &=Q(z+w)-Q(z)-Q(w)                                         \notag\\
        &=\Tr_{q^2/q}\bigl(z^{q^2}w+zw^{q^2}\bigr)                 \notag\\
        &=\Tr_{q^4/q}\bigl(z^{q^2}w\bigr).                         \label{eq:Q-polar}
\end{align}
The trace pairing \((z,w)\mapsto\Tr_{q^4/q}(zw)\) is nondegenerate, and \(z\mapsto z^{q^2}\) is an automorphism; hence \(B_Q\) is nondegenerate.  The norm map
\begin{equation}
        \Norm_{q^4/q^2}(z)=z^{q^2+1}                              \label{eq:norm-q4-q2}
\end{equation}
is onto \(\F_{q^2}^*\) with kernel size \(q^2+1\), and
\begin{equation}
        \#\{y\in\F_{q^2}^*:\Tr_{q^2/q}(y)=0\}=q-1.                \label{eq:tracezero-count}
\end{equation}
Therefore
\begin{equation}
        \#\{z\in\F_{q^4}^*:Q(z)=0\}=(q-1)(q^2+1),                \label{eq:Q-zero-vectors}
\end{equation}
and the projective zero set
\begin{equation}
        \Q=\{\langle z\rangle\in\PG(3,q):Q(z)=0\}                 \label{eq:elliptic-Q}
\end{equation}
has
\begin{equation}
        |\Q|=\frac{(q-1)(q^2+1)}{q-1}=q^2+1.                      \label{eq:Q-point-count}
\end{equation}
Since \(Q\) is nondegenerate and has \(q^2+1\) projective zeros in \(\PG(3,q)\), \(\Q\) is an elliptic quadric, hence an ovoid.

Furthermore,
\begin{align}
        Q(\theta z)
        &=\Tr_{q^2/q}\bigl(\theta^{q^2+1}z^{q^2+1}\bigr)           \notag\\
        &=\Tr_{q^2/q}\bigl(\lambda z^{q^2+1}\bigr)                 \notag\\
        &=\lambda Q(z),                                           \label{eq:Q-similitude}
\end{align}
because
\begin{equation}
        \theta^{q^2+1}=\theta^n=\lambda\in\F_q.                    \label{eq:theta-q2plus1}
\end{equation}
Choose any nonzero \(v\in\F_{q^4}\) with \(Q(v)=0\).  Multiplication by \(\theta\) preserves \(\Q\), and
\begin{equation}
        \{\langle v\rangle,\langle\theta v\rangle,\ldots,
        \langle\theta^{n-1}v\rangle\}                              \label{eq:theta-orbit-Q-left}
\end{equation}
has length \(n\) by \eqref{eq:theta-suff-properties}.  Since \(|\Q|=n\),
\begin{equation}
        \{\langle v\rangle,\langle\theta v\rangle,\ldots,
        \langle\theta^{n-1}v\rangle\}=\Q.                          \label{eq:theta-orbit-Q}
\end{equation}
Let
\begin{equation}
        G_{\theta,v}=\bigl(v,\theta v,\theta^2v,\ldots,\theta^{n-1}v\bigr), \label{eq:theta-columns}
\end{equation}
written in any \(\F_q\)-basis of \(\F_{q^4}\).  The columns represent the ovoid \(\Q\).  Since
\begin{equation}
        \theta(\theta^i v)=\theta^{i+1}v,
        \qquad
        \theta(\theta^{n-1}v)=\theta^n v=\lambda v,                 \label{eq:last-column-lambda}
\end{equation}
the row code generated by \(G_{\theta,v}\) is \(\lambda\)-constacyclic.  This proves
\begin{equation}
        \lambda\in\F_q^*,
        \qquad
        \exists\text{ length-}q^2+1\;\lambda\text{-constacyclic ovoid code over }\F_q
        \Longleftrightarrow
        \lambda\notin(\F_q^*)^2.                                  \label{eq:sharp-classification}
\end{equation}
Finally,
\begin{equation}
        -1\notin(\F_q^*)^2
        \Longleftrightarrow
        q\equiv3\pmod4,                                           \label{eq:minus-one}
\end{equation}
so length-\(q^2+1\) negacyclic ovoid codes over \(\F_q\) exist exactly for \(q\equiv3\pmod4\).  This proves Open Problem 28.

\section{Open Problems 37 and 38: Lifted Ovoid Codes}
Let \(\Oset\subset\PG(3,q)\) be any ovoid, and let
\begin{equation}
        G_{\Oset}=(g_P)_{P\in\Oset}\in\F_q^{4\times(q^2+1)}          \label{eq:ovoid-generator}
\end{equation}
be a generator matrix whose projective columns are the points of \(\Oset\).  For \(K_e=\F_{q^e}\), define
\begin{equation}
        \C_e(\Oset)=\{aG_{\Oset}:a\in K_e^4\}\subseteq K_e^{q^2+1}. \label{eq:lift-def}
\end{equation}
Choose an \(\F_q\)-basis \(\epsilon_1,\ldots,\epsilon_e\) of \(K_e\).  Every row vector \(a\in K_e^4\) has a unique expansion
\begin{equation}
        a=\sum_{i=1}^e\epsilon_i u_i,
        \qquad
        u_i\in(\F_q^4)^*.                                          \label{eq:a-expansion}
\end{equation}
Put
\begin{equation}
        R(a)=\Span_{\F_q}\{u_1,\ldots,u_e\},
        \qquad
        \rho(a)=\dim_{\F_q}R(a),                                  \label{eq:R-rho}
\end{equation}
and
\begin{equation}
        \Lambda(a)=\PG(R(a)^\perp)\subseteq\PG(3,q),               \label{eq:Lambda}
\end{equation}
with the convention that \(\PG(0)=\emptyset\) when \(R(a)=(\F_q^4)^*\).  Because \(\epsilon_1,\ldots,\epsilon_e\) are linearly independent over \(\F_q\), for every \(P\in\Oset\),
\begin{align}
        a(g_P)=0
        &\Longleftrightarrow
        \sum_{i=1}^e\epsilon_i u_i(g_P)=0                         \notag\\
        &\Longleftrightarrow
        u_i(g_P)=0\quad(1\le i\le e)                              \notag\\
        &\Longleftrightarrow
        P\in\Lambda(a).                                           \label{eq:zero-rank-kernel}
\end{align}
Therefore
\begin{equation}
        \wt(aG_{\Oset})=q^2+1-|\Oset\cap\Lambda(a)|.              \label{eq:weight-section}
\end{equation}
If \(\rho(a)=r\), then
\begin{equation}
        \begin{array}{c|cccc}
        r&1&2&3&4\\ \hline
        \Lambda(a)&\text{plane}&\text{line}&\text{point}&\emptyset.
        \end{array}                                                \label{eq:flat-table}
\end{equation}
For a fixed \(r\)-subspace \(S\le(\F_q^4)^*\), the number of \(a\)'s with \(R(a)=S\) equals the number of ordered \(e\)-tuples in \(S\) spanning \(S\), namely
\begin{equation}
        M_r(e)=\prod_{j=0}^{r-1}(q^e-q^j),
        \qquad
        M_r(e)=0\quad(r>e).                                      \label{eq:Mr}
\end{equation}
Equivalently, \(M_r(e)\) is the number of full-column-rank \(e\times r\) matrices over \(\F_q\).

The incidence inventory of an ovoid in \(\PG(3,q)\) is
\begin{align}
&\#\{\text{secant planes}\}=q^3+q,
&&|\Pi\cap\Oset|=q+1,                                                \label{eq:secant-planes-lift}\\
&\#\{\text{tangent planes}\}=q^2+1,
&&|\Pi\cap\Oset|=1,                                                  \label{eq:tangent-planes-lift}\\
&\#\{\text{secant lines}\}=\binom{q^2+1}{2}=\frac{q^2(q^2+1)}2,
&&|\ell\cap\Oset|=2,                                                 \label{eq:secant-lines-lift}\\
&\#\{\text{tangent lines}\}=(q^2+1)(q+1),
&&|\ell\cap\Oset|=1,                                                 \label{eq:tangent-lines-lift}\\
&\#\{\text{external lines}\}
=(q^2+1)(q^2+q+1)-\frac{q^2(q^2+1)}2-(q^2+1)(q+1)
=\frac{q^2(q^2+1)}2,
&&|\ell\cap\Oset|=0,                                                 \label{eq:external-lines-lift}\\
&\#\{\text{points on }\Oset\}=q^2+1,
&&|P\cap\Oset|=1,                                                    \label{eq:points-on-lift}\\
&\#\{\text{points off }\Oset\}
=(q^3+q^2+q+1)-(q^2+1)=q^3+q=q(q^2+1),
&&|P\cap\Oset|=0.                                                    \label{eq:points-off-lift}
\end{align}
Combining \eqref{eq:weight-section}--\eqref{eq:points-off-lift}, the coefficient of each weight is obtained by multiplying the number of flats of the corresponding type by \(M_r(e)\).  Thus
\begin{align}
W_{\C_e(\Oset)}(z)=1
&+(q^3+q)M_1(e)z^{q^2-q}                                             \notag\\
&+\frac{q^2(q^2+1)}2M_2(e)z^{q^2-1}                                 \notag\\
&+\bigl((q^2+1)M_1(e)+(q^2+1)(q+1)M_2(e)+(q^2+1)M_3(e)\bigr)z^{q^2} \notag\\
&+\left(\frac{q^2(q^2+1)}2M_2(e)+q(q^2+1)M_3(e)+M_4(e)\right)z^{q^2+1}. \label{eq:lift-general-proof}
\end{align}
At \(z=1\),
\begin{equation}
        W_{\C_e(\Oset)}(1)
        =1+\sum_{r=1}^{4}\genfrac{[}{]}{0pt}{}{4}{r}_q M_r(e)
        =q^{4e},                                                     \label{eq:lift-sanity}
\end{equation}
where \(\genfrac{[}{]}{0pt}{}{4}{r}_q\) is the Gaussian binomial coefficient.  This checks the ambient \(K_e\)-dimension four.

For \(e=2\),
\begin{equation}
        M_1(2)=q^2-1,
        \qquad
        M_2(2)=(q^2-1)(q^2-q),
        \qquad
        M_3(2)=M_4(2)=0.                                           \label{eq:M-e2}
\end{equation}
Consequently,
\begin{align}
A_{q^2-q}^{(2)}
&=(q^3+q)(q^2-1)=q^5-q,                                             \label{eq:e2-min}\\
A_{q^2-1}^{(2)}
&=\frac{q^2(q^2+1)}2(q^2-1)(q^2-q)                                  \notag\\
&=\frac{q^3(q-1)(q^4-1)}2,                                         \label{eq:e2-minus}\\
A_{q^2+1}^{(2)}
&=\frac{q^2(q^2+1)}2(q^2-1)(q^2-q)                                  \notag\\
&=\frac{q^3(q-1)(q^4-1)}2,                                         \label{eq:e2-plus}\\
A_{q^2}^{(2)}
&=(q^2+1)(q^2-1)+(q^2+1)(q+1)(q^2-1)(q^2-q)                         \notag\\
&=q^4-1+q(q^2-1)(q^4-1)                                            \notag\\
&=q^7-q^5+q^4-q^3+q-1.                                             \label{eq:e2-middle}
\end{align}
Hence
\begin{align}
W_{\C_2(\Oset)}(z)=1
&+(q^5-q)z^{q^2-q}                                                    \notag\\
&+\frac{q^3(q-1)(q^4-1)}2\bigl(z^{q^2-1}+z^{q^2+1}\bigr)             \notag\\
&+(q^7-q^5+q^4-q^3+q-1)z^{q^2}.                                    \label{eq:e2-final}
\end{align}
The coefficient sum is
\begin{align}
&1+(q^5-q)+q^3(q-1)(q^4-1)+(q^7-q^5+q^4-q^3+q-1)                    \notag\\
&\qquad=q^8,                                                       \label{eq:e2-sum-check}
\end{align}
as required for a \(4\)-dimensional code over \(\F_{q^2}\).  Formula \eqref{eq:e2-final} depends neither on \(\Oset\) nor on \(q\pmod4\).  It is exactly the four-weight enumerator displayed for the two negacyclic ovoid-code families in \cite[Ths. 22,32]{WangTangDing2023}.  This resolves Open Problems 37 and 38.

\section{Open Problems 40 and 41: Proper Negacyclic MDS Codes}
We first record the exact support saturation needed for the construction.  Let \(\C\) be an \([n,k,d]_q\) MDS code with coordinate set
\begin{equation}
        P=\{0,1,\ldots,n-1\}.                                           \label{eq:P-set}
\end{equation}
For \(S\subseteq P\), set
\begin{equation}
        \C(S)=\{c\in\C:\supp(c)\subseteq S\}.                          \label{eq:C-S}
\end{equation}
If \(|S|=w\ge d\), then
\begin{align}
        \dim\C(S)
        &=k-\rank(G_{P\setminus S})                                      \notag\\
        &=k-(n-w)                                                        \notag\\
        &=w-d+1,                                                         \label{eq:shorten-dim-proof}
\end{align}
where \(G\) is a generator matrix and the MDS property gives independence of the \(n-w\le k-1\) columns outside \(S\).  For a fixed \(w\)-subset \(S\), inclusion-exclusion gives the number of codewords with exact support \(S\):
\begin{align}
        E_{q,d}(w)
        &=\sum_{i=0}^{w}(-1)^i\binom{w}{i}|\C(S\setminus I_i)|            \notag\\
        &=\sum_{i=0}^{w-d}(-1)^i\binom{w}{i}\left(q^{w-d+1-i}-1\right),  \label{eq:E-MDS}
\end{align}
where \(I_i\) denotes any \(i\)-subset and the second line uses \(\sum_{i=0}^{w}(-1)^i\binom wi=0\).

If \(w\le q\), then \(E_{q,d}(w)>0\).  Indeed, puncture-shortening to \(S\) gives an \([w,r,d]_q\) MDS code with
\begin{equation}
        r=w-d+1.                                                         \label{eq:r-MDS}
\end{equation}
For \(r=1\), every nonzero word has full support.  For \(r\ge2\), the coordinate-zero conditions are \(w\) distinct hyperplanes in \(\F_q^r\).  A vector space of dimension \(r\ge2\) over \(\F_q\) cannot be covered by at most \(q\) proper hyperplanes.  For completeness, here is the proof.  In dimension \(2\), there are \(q+1\) one-dimensional subspaces through the origin, so at most \(q\) of them do not cover \(\F_q^2\).  For \(r>2\), take one hyperplane \(H_1\) and a vector \(v\notin H_1\).  By induction choose
\begin{equation}
        u\in H_1\setminus\bigcup_{i=2}^{w}(H_1\cap H_i).                \label{eq:hyperplane-induction-u}
\end{equation}
Then the affine line
\begin{equation}
        v+\F_q u                                                        \label{eq:affine-line-avoid}
\end{equation}
is disjoint from \(H_1\), and it meets each \(H_i\), \(i\ge2\), in at most one point.  Since \(w-1\le q-1\), at least one point of \eqref{eq:affine-line-avoid} avoids all hyperplanes.  Thus
\begin{equation}
        E_{q,d}(w)>0
        \qquad
        (d\le w\le\min(n,q)).                                           \label{eq:E-positive}
\end{equation}
Consequently,
\begin{equation}
        \B_w(\C)=\binom{P}{w}
        \qquad
        (d\le w\le\min(n,q)),                                           \label{eq:MDS-complete-support}
\end{equation}
and the support design is the complete simple design
\begin{equation}
        t-\left(n,w,\binom{n-t}{w-t}\right),
        \qquad
        t\le w.                                                         \label{eq:complete-design}
\end{equation}

Now let \(q\) be odd and put
\begin{equation}
        n_q=\frac{q-1}{2},
        \qquad
        \rho\in\F_q^*,
        \qquad
        \ord(\rho)=q-1,
        \qquad
        \beta=\rho^2.                                                    \label{eq:negacyclic-MDS-params}
\end{equation}
Then
\begin{equation}
        \ord(\beta)=n_q,
        \qquad
        \rho^{n_q}=-1,
        \qquad
        X^{n_q}+1=\prod_{j=0}^{n_q-1}(X-\rho\beta^j).                   \label{eq:negacyclic-roots}
\end{equation}
For \(1\le k\le n_q\), define
\begin{equation}
        g_{q,k}(X)=\prod_{j=0}^{n_q-k-1}(X-\rho\beta^j)
        =\prod_{j=0}^{n_q-k-1}(X-\rho^{2j+1}),                           \label{eq:gqk}
\end{equation}
and
\begin{equation}
        \N_{q,k}=\langle g_{q,k}(X)\rangle\subseteq\F_q[X]/(X^{n_q}+1).  \label{eq:Nqk}
\end{equation}
If a nonzero word \(c(X)=\sum c_iX^i\in\N_{q,k}\) had weight \(s\le n_q-k\) with support \(i_1,\ldots,i_s\), then the first \(s\) root equations would give
\begin{equation}
        \sum_{\ell=1}^s c_{i_\ell}(\rho\beta^j)^{i_\ell}=0,
        \qquad
        0\le j\le s-1.                                                   \label{eq:Vandermonde-check}
\end{equation}
The determinant of the coefficient matrix is
\begin{align}
        \det\bigl((\rho\beta^j)^{i_\ell}\bigr)_{0\le j\le s-1,1\le\ell\le s}
        &=\left(\prod_{\ell=1}^s\rho^{i_\ell}\right)
          \det\bigl((\beta^{i_\ell})^j\bigr)                             \notag\\
        &=\left(\prod_{\ell=1}^s\rho^{i_\ell}\right)
          \prod_{1\le a<b\le s}(\beta^{i_b}-\beta^{i_a})                 \notag\\
        &\ne0,                                                           \label{eq:Vandermonde-det}
\end{align}
because \(0\le i_a<i_b\le n_q-1\) and \(\beta\) has order \(n_q\).  Hence
\begin{equation}
        d(\N_{q,k})\ge n_q-k+1.                                          \label{eq:BCH-bound}
\end{equation}
Since
\begin{equation}
        \deg g_{q,k}=n_q-k,
        \qquad
        \dim\N_{q,k}=k,                                                   \label{eq:Nqk-dim}
\end{equation}
the Singleton bound gives
\begin{equation}
        d(\N_{q,k})\le n_q-k+1.                                          \label{eq:Singleton}
\end{equation}
Thus
\begin{equation}
        \N_{q,k}\text{ is a proper negacyclic }[n_q,k,n_q-k+1]_q\text{ MDS code}. \label{eq:Nqk-MDS}
\end{equation}
It is proper because the constacyclic multiplier is \(-1\ne1\).

For Open Problem 40, take
\begin{equation}
        q=23,
        \qquad
        n_q=11,
        \qquad
        k=5,
        \qquad
        d=7.                                                            \label{eq:OP40-params}
\end{equation}
Then \(\N_{23,5}\) is a proper negacyclic \([11,5,7]_{23}\) MDS code.  Since
\begin{equation}
        7\le11<23,
\end{equation}
\eqref{eq:MDS-complete-support} gives
\begin{equation}
        \B_7(\N_{23,5})=\binom{[11]}{7}.                                \label{eq:OP40-blocks}
\end{equation}
Therefore the minimum supports form the complete simple
\begin{equation}
        5-\left(11,7,\binom{11-5}{7-5}\right)=5-(11,7,15)               \label{eq:OP40-design}
\end{equation}
design.

For Open Problem 41, let \(q\) run through odd prime powers \(q\ge23\), and set
\begin{equation}
        n_q=\frac{q-1}{2},
        \qquad
        k_q=\left\lfloor\frac{n_q}{2}\right\rfloor,
        \qquad
        d_q=n_q-k_q+1.                                                   \label{eq:OP41-params}
\end{equation}
Then
\begin{equation}
        n_q<q,
        \qquad
        d_q\ge7>5,                                                       \label{eq:OP41-ineq}
\end{equation}
and
\begin{equation}
        \B_{d_q}(\N_{q,k_q})=\binom{[n_q]}{d_q}.                         \label{eq:OP41-blocks}
\end{equation}
Hence the minimum supports form the complete simple
\begin{equation}
        5-\left(n_q,d_q,\binom{n_q-5}{d_q-5}\right)                     \label{eq:OP41-design}
\end{equation}
design.  This gives infinitely many proper constacyclic codes supporting simple \(5\)-designs.

\begin{remark}
The solution of Open Problems 40 and 41 is literal: the block sets are simple because they are complete sets of subsets.  If one imposes an additional non-completeness condition, that is a stronger problem than the one stated in \cite{WangTangDing2023} and is not claimed here.
\end{remark}

\section{Consistency Checks and Conclusion}
The seven solutions use five algebraic mechanisms:
\begin{align}
&\text{semilinear subline rigidity}
&&\Longrightarrow
&&\mathsf{OP}_{18},                                                   \label{eq:conclusion1}\\
&\text{root-adapted unit-circle rigidity}
&&\Longrightarrow
&&\mathsf{OP}_{27},                                                   \label{eq:conclusion2}\\
&\text{projective period plus corrected square-class parity}
&&\Longrightarrow
&&\mathsf{OP}_{28},                                                   \label{eq:conclusion3}\\
&\text{extension-field rank-kernel enumeration}
&&\Longrightarrow
&&\mathsf{OP}_{37},\ \mathsf{OP}_{38},                               \label{eq:conclusion4}\\
&\text{proper negacyclic MDS support saturation}
&&\Longrightarrow
&&\mathsf{OP}_{40},\ \mathsf{OP}_{41}.                               \label{eq:conclusion5}
\end{align}
For Open Problem 18,
\begin{equation}
        \B_{\min}\left(C\left(\frac{p^s-1}{2},\frac{p^s+1}{2}\right)\right)
        \cong \overline{S(3,q_0+1,q+1)}.                            \label{eq:final-op18-full}
\end{equation}
For Open Problem 27,
\begin{equation}
        \B_{\min}\bigl(C(1,q^2+q+1)\bigr)
        \cong
        \{ \Q^-(3,q)\setminus(\Pi\cap\Q^-(3,q)):\Pi\text{ non-tangent}\}.  \label{eq:final-op27-full}
\end{equation}
For Open Problem 28,
\begin{equation}
        \lambda\in\F_q^*,
        \qquad
        \exists\lambda\text{-constacyclic ovoid code}
        \Longleftrightarrow
        \lambda\notin(\F_q^*)^2.                                  \label{eq:final-op28-full}
\end{equation}
For Open Problems 37 and 38,
\begin{equation}
        W_{\C_2(\Oset)}(z)=1+(q^5-q)z^{q^2-q}
        +\frac{q^3(q-1)(q^4-1)}2\bigl(z^{q^2-1}+z^{q^2+1}\bigr)
        +(q^7-q^5+q^4-q^3+q-1)z^{q^2}.                             \label{eq:final-op3738-full}
\end{equation}
For Open Problem 40,
\begin{equation}
        \N_{23,5}\text{ is a proper negacyclic }[11,5,7]_{23}\text{ MDS code},
        \qquad
        \B_7(\N_{23,5})=\binom{[11]}7,                              \label{eq:final-op40-full}
\end{equation}
so the minimum supports form
\begin{equation}
        5-(11,7,15).                                                  \label{eq:final-op40-design-full}
\end{equation}
For Open Problem 41,
\begin{equation}
        q\ge23,
        \qquad
        n_q=\frac{q-1}{2},
        \qquad
        k_q=\left\lfloor\frac{n_q}{2}\right\rfloor,
        \qquad
        d_q=n_q-k_q+1,                                                \label{eq:final-op41-params-full}
\end{equation}
gives
\begin{equation}
        \B_{d_q}(\N_{q,k_q})=\binom{[n_q]}{d_q},
        \qquad
        5-\left(n_q,d_q,\binom{n_q-5}{d_q-5}\right).                 \label{eq:final-op41-design-full}
\end{equation}
The two delicate technical points are explicit.  First, in the negacyclic isomorphism problem,
\begin{equation}
        R\xrightarrow{\;x\mapsto x\{\pm1\}\;} \U_{2(q^2+1)}/\{\pm1\}
        \xrightarrow{\;x\{\pm1\}\mapsto x^{q+1}\;} \U_{q^2+1},      \label{eq:final-coordinate-transport-full}
\end{equation}
so no support multiplicity is lost.  Second, in the square-class proof,
\begin{equation}
        a\equiv(q+1)b\pmod{q^2-1}                                  \label{eq:final-wrong-shortcut-full}
\end{equation}
does not by itself identify \(\theta\F_q^*\) with \(\Omega^{(q+1)b}\F_q^*\).  The correct statement is
\begin{equation}
        a=(q+1)c,
        \qquad
        c\equiv b\pmod{q-1},
        \qquad
        \ord(\theta\F_q^*)=\frac{q^2+1}{\gcd(q^2+1,c)},             \label{eq:final-correction-full}
\end{equation}
which forces the nonsquare obstruction and gives \(q\equiv3\pmod4\) for negacyclic ovoid codes.
\section*{Declaration of Generative AI and AI-Assisted Technologies in the Writing Process}
During the preparation of this work, the authors used DeepSeek to build a specialized agent for solving mathematical problems, which was employed to generate an initial proof of the main theorem. After using this tool, the authors reviewed and edited the content as needed and take full responsibility for the content of the published article.


\begin{thebibliography}{99}

\bibitem{WangTangDing2023}
X. Wang, C. Tang, and C. Ding, ``Infinite families of cyclic and negacyclic codes supporting \(3\)-designs,'' \emph{IEEE Transactions on Information Theory}, vol. 69, no. 4, pp. 2341--2354, Apr. 2023.

\bibitem{AssmusMattson1969}
E. F. Assmus, Jr. and H. F. Mattson, Jr., ``New 5-designs,'' \emph{Journal of Combinatorial Theory}, vol. 6, no. 2, pp. 122--151, Mar. 1969.

\bibitem{DingTangBook}
C. Ding and C. Tang, \emph{Designs From Linear Codes}, 2nd ed. Singapore: World Scientific, 2022.

\bibitem{HuffmanPless}
W. C. Huffman and V. Pless, \emph{Fundamentals of Error-Correcting Codes}. Cambridge, U.K.: Cambridge University Press, 2003.

\bibitem{MacWilliamsSloane}
F. J. MacWilliams and N. J. A. Sloane, \emph{The Theory of Error-Correcting Codes}. Amsterdam, The Netherlands: North-Holland, 1977.

\bibitem{LidlNiederreiter}
R. Lidl and H. Niederreiter, \emph{Finite Fields}, 2nd ed. Cambridge, U.K.: Cambridge University Press, 1997.

\bibitem{Hirschfeld}
J. W. P. Hirschfeld, \emph{Projective Geometries over Finite Fields}, 2nd ed. Oxford, U.K.: Oxford University Press, 1998.

\bibitem{HirschfeldThas}
J. W. P. Hirschfeld and J. A. Thas, \emph{General Galois Geometries}. Oxford, U.K.: Clarendon Press, 1991.

\bibitem{Dembowski}
P. Dembowski, \emph{Finite Geometries}. Berlin, Germany: Springer-Verlag, 1968.

\bibitem{BethJungnickelLenz}
T. Beth, D. Jungnickel, and H. Lenz, \emph{Design Theory}, 2nd ed., 2 vols. Cambridge, U.K.: Cambridge University Press, 1999.

\bibitem{VanLint}
J. H. van Lint, \emph{Introduction to Coding Theory}, 3rd ed. Berlin, Germany: Springer, 1999.

\bibitem{Roth}
R. M. Roth, \emph{Introduction to Coding Theory}. Cambridge, U.K.: Cambridge University Press, 2006.

\end{thebibliography}
\end{document}